\documentclass[unsortedaddress,nofootinbib]{revtex4-1}
\usepackage{amsmath,empheq,color}
\usepackage{amsfonts}
\usepackage{amsthm}
\usepackage{amssymb}
\usepackage{graphicx}
\usepackage{hyperref}
\usepackage[normalem]{ulem}

	\newcommand{\ncd}{\newcommand}
	\ncd{\mrm}    {\mathrm}
	\ncd{\beq} {\begin{equation}}
	\ncd{\eeq} {\end{equation}}

	\def\d{{\rm d}}

	\newtheorem{Theorem}{Theorem}

\begin{document}

	\title{Liouville's Theorem and the canonical measure for nonconservative systems from contact geometry}

\author{A. Bravetti}
        \email{bravetti@correo.nucleares.unam.mx}
	\affiliation{Instituto de Ciencias Nucleares, Universidad Nacional Aut\'onoma de M\'exico,\\ AP 70543, M\'exico, DF 04510, Mexico.}

	\author{D. Tapias}
	\email{diego.tapias@nucleares.unam.mx }
	\affiliation{Facultad de Ciencias, Universidad Nacional Aut\'onoma de M\'exico,\\ AP 70543, M\'exico, DF 04510, Mexico.}

	\begin{abstract}
Standard statistical mechanics of conservative systems relies on the symplectic geometry of the phase space. This is exploited to derive Hamilton's equations, Liouville's theorem and to find the canonical invariant measure. In this work we analyze the statistical mechanics of a class of nonconservative systems stemming from contact geometry. In particular, we find out the generalized Hamilton's equations, Liouville's theorem and the microcanonical and canonical measures invariant under the contact flow. Remarkably, the latter measure has a power law density distribution with respect to the standard contact volume form. Finally, we argue on the several possible applications of our results.
	\end{abstract}



\maketitle


\section{Introduction}
	Statistical mechanics is one of the most powerful tools for the investigation of the {collective properties} of large systems and therefore it has been a 
	great success in all fields of natural and human sciences.
	Nevertheless, the full formal understanding of statistical mechanics can only be achieved for conservative systems, i.e.
	systems whose description can be given in terms of Hamilton's equations.
	Such construction is based on Hamilton's principle, symplectic geometry, Liouville's theorem and Gibbs canonical measure \cite{TuckermanBook}.
	However, only isolated systems are conservative and therefore most of the interesting systems evade the standard treatment.
	A central problem for both the theoretical development of statistical mechanics and its numerical implementations 
	is that of finding a general description of the dynamics of {nonconservative systems}, including conditions for 
	the existence of invariant measures. 
	 {In all cases} one is faced with two major {issues}. First,
	to derive the equations of motion and second, to prove the existence of an invariant measure along the flow.

{Over the last decades, there have been many attempts to extend the statistical mechanics of standard symplectic Hamiltonian systems. 
One of the main motivations is the description of the dynamics of equilibrium ensembles different from the microcanonical one. In order to accomplish this task, 
non-Hamiltonian equations of motion have been proposed which mimic different kinds of ensembles by means of the so-called \emph{extended systems}.
Such dynamics has been extensively used to perform numerical simulations (see e.g. \cite{1980JChPh722384A,mukunda1974classical,hoover1985canonical,1984Evans,1998Morriss,evansbook} 
for some relevant references). Later, it was shown that the structure of the equations of motion of such extended 
systems can be framed within an algebraic scheme through the use of generalized Poisson brackets \cite{sergi2001non,sergi2003non,SergiGiaquinta}.
Moreover, another interesting perspective on non-Hamiltonian systems has been proposed, based on a more geometric-oriented reasoning \cite{Tuckerman2001, Ezra2004, TarasovJPA}.}

%
	
	{Here we focus on the geometric approach, since this is a natural setting for the construction of invariant tools, 
	as geometric objects are, by construction, coordinate independent.
	In particular, the objective of this work is the extension of the statistical mechanics of conservative Hamiltonian systems
	to a class 
	of nonconservative systems stemming from contact geometry.
	 This choice is motivated by different reasons. First of all contact Hamiltonian dynamics is the most natural generalization of symplectic dynamics \cite{arnold2001dynamical}.}
	 Moreover, {we are prompted by} the analogy with thermodynamic systems, 
	 whose phase space is a contact manifold \cite{Mrugaa:2000aa,BoyerCompletelyIntegrable,2014Bravb,
	 2014Bravc}.
	 Besides, recent works suggest
	 that numerical techniques for Monte Carlo simulations can be improved by  using  contact flows  \cite{2014arXiv1405.3489B}. 
	 Finally, 
	 although it is most common to formulate classical mechanics in terms of Hamilton's
	 equations and symplectic geometry, 
	 nevertheless the description by means of the Hamilton-Jacobi equation results in a contact phase space \cite{Rajeev2008768}. 
	 {For all these reasons, here we provide a number of results for the statistical mechanics of
	 the class of nonconservative contact Hamiltonian systems.
	Our} perspective presents several desirable properties.
	First of all, the equations of motion are derived from the geometric picture as Hamiltonian flows (in the contact context), as in the symplectic case.
	Secondly, for all such systems we  {show that there exist both a uniform invariant measure along the orbits of the flow in which the contact Hamiltonian is conserved and
	a distinct invariant measure along the orbits in which it is not conserved. Remarkably, the latter measure} depends only on the Hamiltonian,
	 thus being a generalization of Gibbs' canonical measure. 	 
	 In the third place, the probability distribution associated with this measure with respect to the standard contact volume form
	 is a power law, which is the distribution encountered everywhere in nature \cite{powerlawreview}. 
	Finally,  the contact formulation allows for an elegant and natural understanding {of  nonconservative
	 contact Hamiltonian systems} as a generalization of the conservative case, recovering symplectic dynamics as a special case.

\section{Symplectic geometry and the phase space of conservative systems}
{Conservative systems} are dynamical systems for which the {mechanical} energy is conserved.
Therefore, their description can be given in terms of the Hamiltonian function, which gives 
Hamilton's equations of motion in the phase space. 
In particular, the phase space of a conservative system
is the cotangent bundle of the configuration manifold, that is
{a $2n$-dimensional manifold} $\Gamma$ coordinatized by the particles' generalized coordinates and momenta $q^{a}$ and $p_{a}$, with 
$a=1,\dots,n$. 
Such manifold is naturally endowed
with a $1$-form 
	\beq\label{taut}
	\alpha=p_{a}\d q^{a}\,, 
	\eeq
{where here} and in the following Einstein's summation convention over repeated indices is assumed.
The exterior derivative of $\alpha$ defines the standard symplectic form on $\Gamma$, that is
	\beq\label{Omega}
	\Omega=\d\alpha=\d p_{a}\wedge \d q^{a}\,.
	\eeq
Given the Hamiltonian function $H(q^{a},p_{a})$ on $\Gamma$, Hamilton's equations of motion
follow from 
	\beq\label{Hameq1}
	-\d H = \Omega(X_{H},\cdot)\,,
	\eeq 
where $X_{H}$ is the \emph{Hamiltonian vector field} defining the evolution of the system.
Therefore, the equations of motion take the standard form
	\beq\label{Hameq2}
	\dot q^{a}=\frac{\partial H}{\partial p_{a}} \qquad \dot p_{a}=-\frac{\partial H}{\partial q^{a}}\,.
	\eeq
A system whose evolution is governed by \eqref{Hameq2} is usually called a \emph{Hamiltonian system}.
However, in this work we want to generalize the notion of a Hamiltonian system to comprehend the case of 
Hamiltonian systems in contact geometry.
Therefore we will refer to the dynamic system given by \eqref{Hameq2} as a \emph{conservative system}.

The crucial problem in statistical mechanics is to find an {invariant measure} on the phase space for the flow $\phi$ associated to the equations of motion.
Geometrically, an \emph{invariant measure} $\mu$ is 
a volume form on $\Gamma$
such that 
	\beq
	\phi_{t}^{*}\left(\mu_{t}\right)=\mu
	\eeq 
for any $t$, where $\phi_{t}^{*}$ represents the pullback induced by the flow\footnote{{In usual statistical applications based on symplectic Hamiltonian flow the invariant measure is most naturally defined in terms of the pushforward. However,
 in this work we consider a dynamics based on the contact Hamiltonian flow (c.f. Section \ref{RCG}), which can exhibit fixed points that
 obstruct the pushforward \cite{2014arXiv1405.3489B}. 
 Therefore a definition by means of the pullback seems more appropriate in this case.}}.
We always assume that the measure on $\Gamma$ is given in terms of a probability density $\rho(q^{a},p_{a})$, which means
that $\d\mu$ can be written as $\d\mu=\rho(q^{a},p_{a})\d x$, where $\d x$ is a short form to indicate the 
 volume element of $\Gamma$.
In the usual statistical mechanics of conservative systems, the phase space is a symplectic manifold 
equipped with the standard volume element $\Omega^{n}$, {where $\Omega$ is given by} \eqref{Omega}.
 Liouville's theorem can be written in a compact and geometric form as
	 \beq\label{LiouvilleTh}
	 \pounds_{X_{H}}\Omega=0\,,
	 \eeq
 where $\pounds_{X_{H}}$ denotes the Lie derivative with respect to the symplectic flow $X_{H}$ associated with the Hamiltonian $H$ - c.f. eq. \eqref{Hameq2}. 
 This means  that the volume element $\Omega^{n}$ is invariant  along any symplectic flow.
 Therefore $\Omega^{n}$ is a natural measure on any symplectic manifold and because the space of measures on manifolds is one-dimensional, any other measure must be proportional,
 for some proportionality function $\rho(q^{a},p_{a})$.
Moreover, the evolution of any function is governed by the Poisson brackets. 
Therefore, one can also find the conditions for the probability density $\rho(q^{a},p_{a})$ to be invariant along the flow of $X_{H}$, resulting in Liouville's equation
	\beq\label{LiouvilleEq}
	\frac{\d \rho}{\d t}=\frac{\partial \rho}{\partial t}+\{\rho,H\}
	=0\,.
	\eeq
Solutions to eq. \eqref{LiouvilleEq} can be found easily. For example, 
if $\rho$ is positive, integrable and depends only on the Hamiltonian, then the resulting measure $\d\mu\equiv\rho(H)\Omega^{n}$
is an invariant of the flow.
 This fact is of central importance in statistical mechanics, because it guarantees 
 - {together with other hypotheses such as \emph{ergodicity} or the more recent concept of \emph{typicality} \cite{goldstein2006canonical,goldstein2004boltzmann}} -
 that one can perform measures along the evolution of the system and
 exchange time averages with ensemble averages, which is the starting point of all statistical mechanical calculations \cite{TuckermanBook}.

\section{Nonconservative contact Hamiltonian systems}
Nonconservative systems are ubiquitous in nature. 
In fact, strictly speaking, only abstract isolated systems for which one can write down completely the microscopic dynamics are conservative \cite{GalleyPRL}. 
Therefore, the majority of systems cannot be thought as conservative.
When the system is nonconservative, its dynamics cannot be given in terms of a standard mechanical Hamiltonian function as in \eqref{Hameq2}.
As an example, we consider here 
the simplest dissipative system, whose dynamical equations can be given as
	\beq\label{dissipative}
	\dot q^{a}=\frac{\partial H}{\partial p_{a}} \qquad \dot p_{a}=-\frac{\partial H}{\partial q^{a}}-\alpha p_{a}.
	\eeq 
The term $\alpha p_{a}$ in the second equation is obviously a nonconservative force (it is a dissipative term), 
and hence it cannot be derived from a Hamiltonian function  
 in the standard symplectic picture of the phase space. This is the reason why nonconservative systems are usually referred to as \emph{non-Hamiltonian systems}.
 However, we will propose in the next section a Hamiltonian formulation for a large class of such systems, including the basic example provided in \eqref{dissipative}, 
 by simply assuming that  
{the phase space in this case has to be a contact manifold}. 
Hence, we will simply refer to such systems as {\emph{contact Hamiltonian systems} \cite{BoyerCompletelyIntegrable}}. 
It is important to remark at this point the two most relevant problems {for the formulation} of the statistical mechanics of nonconservative systems:
	\begin{itemize}
	\item[i)] the equations of motion have to be provided by some other means, due to the fact that dissipative terms cannot be achieved from Hamilton's equations \eqref{Hameq2}.
	\item[ii)] 	The divergence of the flow, defined by the relation \cite{Ezra2004}
			\beq\label{divergence}
			\pounds_{X}\Omega=\left(\text{div}_{\Omega}X\right)\Omega\,,
			\eeq 
	 in this case does not vanish, and therefore the standard Liouville
	theorem does not apply. Thus one needs to prove for any such flow the existence of an appropriate invariant measure.
	\end{itemize} 

{In the following we will provide new geometric results that are useful for the statistical mechanics of contact Hamiltonian systems.}
 We will show that this picture can be seen as a natural generalization of the classical (symplectic)
formulation of the phase space of conservative systems. Moreover, our formulation automatically resolves the problems i) and ii) stated  above.

\subsection{Review of contact Hamiltonian systems}\label{RCG}
Let us  
start now by reviewing briefly some concepts of 
contact geometry that will be useful later {(see e.g. \cite{arnold2001dynamical,Mrugaa:2000aa,BoyerCompletelyIntegrable,2014Bravc,2014Bravb} for more details).}
A \emph{contact manifold} $\mathcal T$ is a {$(2n+1)$-dimensional} manifold endowed with a 1-form $\eta$ that 
satisfies the non-integrability condition
\begin{equation}\label{integraxx}
   \eta \wedge (\d \eta)^n \neq 0\,.
\end{equation}
The left hand side in \eqref{integraxx} thus provides the \emph{standard volume form} on $\mathcal T$, analogously to \eqref{Omega} for the symplectic case.
Associated to $\eta$ there is always a  global vector field $\xi$ --
the \emph{Reeb vector field} -- defined uniquely by the two conditions
\beq\label{Reeb}
\eta(\xi)=1 \quad \text{and} \quad \d \eta (\xi,\cdot)=0\,.
\eeq
Now we want to define the dynamics in the phase space $\mathcal T$. 
{Using Cartan's identity 
	\beq\label{cartan}
	\pounds_{X_{h}}\eta=\d \eta(X_{h},\cdot)+\d[\eta(X_{h})](\cdot)
	\eeq
and}
the $1$-form $\eta$,
we can associate to every
differentiable function $h:\mathcal{T} \rightarrow \mathbb{R}$, 
 a vector field $X_{h}$, called the \emph{Hamiltonian vector field generated by $h$},
defined through the relation
	\beq\label{isomorphismLiealg}
	h = \eta\left(X_h\right)\,,
	\eeq
and we say that $h$ is a \emph{contact Hamiltonian} \cite{arnold2001dynamical,Mrugaa:2000aa,BoyerCompletelyIntegrable}.
{Equations \eqref{cartan} and \eqref{isomorphismLiealg} define the contact Hamiltonian dynamics and are the analogue of \eqref{Hameq1} in the symplectic case.}

 Using  the above identification between vector fields and functions on $\mathcal T$, one can define the \emph{Jacobi
 brackets}
 	\beq\label{Jacobibr}
	\left\{\eta(X),\eta(Y)\right\}_{\eta}=\eta\left([X,Y]\right)
	\eeq   
which give a Lie algebra structure to functions over $\mathcal T$ and
 are the contact analogue of the Poisson brackets of symplectic geometry  \cite{2014Bravb}.
When the 1-form $\eta$ defining the contact structure and the Hamiltonian function $h$ are fixed on $\mathcal T$,
we say that the quadruple 
$(\mathcal T,\mathcal D, \eta,h)$ 
is a \emph{contact Hamiltonian system} \cite{BoyerCompletelyIntegrable}.

Additionally, it is always possible to find a set of local (Darboux) coordinates $(S,q^{a},p_{a})$ for $\mathcal{T}$ such that the 1-form $\eta$ can be written as
	\begin{equation}\label{1stform}
	\eta = \d S + p_a \d q^a\,.
	\end{equation}

The Reeb vector field is given in local Darboux coordinates by $\xi=\frac{\partial}{\partial S}$ 
and  generates a natural splitting of the tangent bundle, that is 
\beq\label{splitting}
T\mathcal T = V_{\xi}\oplus \mathcal{D}\, ,
\eeq
where $V_{\xi}$ is the \emph{vertical} sub-space generated by $\xi$ and $\mathcal D$ is the \emph{horizontal} (\emph{contact})  
distribution given by $\mathcal D={\rm ker} \eta$.
It is always possible to find locally a basis of the tangent space which is adapted to the splitting  \eqref{splitting}, given by
\cite{2014Bravb}
	\beq\label{Dbasis}
	\Big{\{}\xi,\hat{P}^{i},\hat{Q}_{i}\Big{\}}=\left\{\frac{\partial}{\partial S},\frac{\partial}{\partial p_i},  p_i \frac{\partial}{\partial S}-\frac{\partial}{\partial q^i} \right\}\,.
	\eeq
{Remarkably,  the vectors} of such basis satisfy the commutation relations
	\beq
	\label{halgebra}
	[\hat P^i,\hat Q_j] = \delta^i_{\ j} \xi, \quad [\xi,\hat P^i] = 0 \quad \text{and} \quad [\xi, \hat Q_i] = 0,
	\eeq
{and therefore the contact phase space is locally isomorphic to} the $n$th Heisenberg group \cite{2014Bravb}. 


{Let $(\mathcal T,\mathcal D, \eta,h)$ be a contact Hamiltonian system.}
In local Darboux coordinates the Hamiltonian vector field $X_{h}$ takes the form 	
	\beq
	\begin{split}
	\label{generic.ham}
	X_h =& \left( h - p_a \frac{\partial h}{\partial p_a}\right)\frac{\partial}{\partial S} \\
		+& \left(p_a \frac{\partial h}{\partial S}-\frac{\partial h}{\partial q^a} \right)\frac{\partial }{\partial p_a}\\ 
		+& \left(\frac{\partial h }{\partial p_a} \right)\frac{\partial }{\partial q^a}\,.
	\end{split}
	\eeq	
{In terms of the adapted basis introduced in \eqref{Dbasis}}, we can express the action of $X_{h}$ on a function $f$ as
	\beq\label{XhJacobi}
	X_{h}f=h\,\xi(f) +\hat Q_{a}(h)\hat P^{a}(f)-\hat P^{a}(h)\hat Q_{a}(f)\,.
	\eeq
We say that a function $f\in C^{\infty}(\mathcal T)$ is a \emph{first integral} of the contact Hamiltonian system $(\mathcal T,\mathcal D, \eta,h)$ if
$f$ is constant along the flow of $X_{h}$, that is if $X_{h}f=0$.
From the above equation \eqref{XhJacobi} it follows in general that 
	\beq\label{flowofh}
	X_{h}h=h\,\xi(h)\,,
	\eeq
and therefore the Hamiltonian function itself is not in general a first integral
of its flow. Indeed $h$ is a first integral if and only if it is a \emph{basic function}, i.e. $\xi(h)=0$.
Finally, it is worth noting that eq. \eqref{XhJacobi}
implies that for every function $f$ that depends only on $h$ 
	\beq\label{evolutionfh}
	X_{h}f(h)=h\,\xi\left(f(h)\right)=h\,f'(h)\xi(h)
	\eeq 
and as a consequence, any $f(h)$ in general is constant only along the flow lines with $h = 0$.

According to equation \eqref{generic.ham}, the flow of $X_{h}$ can be explicitly written in Darboux coordinates as
	\begin{empheq}[left=\empheqlbrace]{align}
	\label{z1}
	& \dot S \,\,\,= h - p_a \frac{\partial h}{\partial p_a}\,,\\
	\label{z2}
	& \dot{p}_{a} \,=  -\frac{\partial h}{\partial q^a} + p_a \frac{\partial h}{\partial S}\,,\\
	\label{z3}
	& \dot{q}^{a} = \frac{\partial h }{\partial p_a} \,,
	\end{empheq}
where its similarity with Hamilton's equations of symplectic mechanics is manifest - c.f. eq. \eqref{Hameq2}.
In fact, these are the generalization of Hamilton's equations to a contact manifold. 
In particular, when $h$ is a basic function
 equations \eqref{z2} and \eqref{z3} give exactly Hamilton's equations. 
Finally, let us note that \eqref{z1} is an extra equation for the evolution of the variable $S$. Such equation
can be rewritten by means of \eqref{z3} as $\dot S=h-p_{a}\dot{q}^{a}$, suggesting that $S$ is a generalization
of Hamilton's principal function to the contact case.
Eqs. \eqref{z1}-\eqref{z3} generalize the symplectic equations \eqref{Hameq2} and therefore can include
a large class of models, such as e.g. the basic dissipative systems in \eqref{dissipative} 
or the more sophisticated `thermostatted dynamics' \cite{TuckermanBook,evansbook}.  
However, the difference with previous proposals is that in our approach the equations of motion follow
from a (contact) Hamiltonian.

\subsection{Invariant distributions for nonconservative contact Hamiltonian systems}\label{LiouvTh}
Now let us apply the machinery of contact Hamiltonian flows to derive
the invariant measure {for this class of nonconservative systems.} 
From eq. \eqref{generic.ham}, it is easy to see that 
	\beq
	\pounds_{X_{h}}\eta=\xi(h)\,\eta
	\eeq
and that 
	\beq\label{LieDerVolume}
	\pounds_{X_{h}}\left(\eta\wedge(\d\eta)^{n}\right)=(n+1)\xi(h)\,\left(\eta\wedge(\d\eta)^{n}\right).
	\eeq

Therefore, recalling the definition of the divergence of a flow - eq. \eqref{divergence} - we find from eq. \eqref{LieDerVolume}
that the divergence of any contact Hamiltonian system {with respect to the standard contact volume form $\eta\wedge (\d\eta)^{n}$} is given by
	\beq\label{contactdivergence}
	\text{div}X_{h}=(n+1)\xi(h)\,.
	\eeq
The first consequence of this expression is that we can easily recover the standard conservative case. In fact, we can think of a conservative
system in this more general formalism as a system for which $h=H(q^{a},p_{a})$. 
This means that we are assuming that the Hamiltonian describing the system is a standard symplectic Hamiltonian. 
Therefore, in such case the system \eqref{z1}-\eqref{z3}
reads
	\begin{empheq}[left=\empheqlbrace]{align}
	\label{z1bis}
	& \dot S \,\,\,= H(q^{a},p_{a})-p_{a}\,\dot{q}^{a}\,,\\
	\label{z2bis}
	& \dot{p}_{a} \,=  -\frac{\partial H}{\partial q^a} \,,\\
	\label{z3bis}
	& \dot{q}^{a} = \frac{\partial H }{\partial p_a} \,.
	\end{empheq}
{Eqs. \eqref{z2bis} and \eqref{z3bis} are the same as \eqref{Hameq2}, while \eqref{z1bis} states that $S$ coincides with Hamilton's
principal function in the conservative case}. 
Moreover, since $h=H(q^{a},p_{a})$ does not depend on $S$, we have $\xi(h)=0$ and from 
eq. \eqref{contactdivergence} 
 the divergence of the flow on the phase space vanishes, which is the basic point for the proof of the standard Liouville theorem.
{However, for a general  contact Hamiltonian system, 
the function $h$ leads to a flow with a non-vanishing divergence as in \eqref{contactdivergence}.
Therefore, we state here  two general theorems, which are the main results of this work and can be seen as the generalization of Liouville's theorem to contact Hamiltonian dynamics.

Before stating the theorems, let us remark that from eq. \eqref{flowofh} it follows that  in general of all the possible level surfaces $h = c$, the flow $X_{h}$ is tangent only to $h=0$.
This implies that the orbits of the contact Hamiltonian flow split the manifold into three disconnected regions, corresponding to $h>0$,
$h=0$ and $h<0$ respectively. Accordingly, we will find two different invariant measures, one corresponding to the flow with $h=0$ and the other to the flow with $h\neq0$.

Let us begin with the level surface $h^{-1}(0)$. The following theorem provides an invariant measure for orbits of the contact flow in this surface. Moreover, we show that
it is a uniform (microcanonical) distribution.
}

	\begin{Theorem}[Microcanonical invariant measure on $h^{-1}(0)$]\label{Theorem0}
	For any nonconservative system \eqref{z1}-\eqref{z3} given by the corresponding contact Hamiltonian $h(S,q^{a},p_{a})$, the measure
	\beq\label{MicroInvariantMeasure}
	\d\mu|_{h^{-1}(0)}\equiv \left.\frac{f(h)}{\mathcal Z_{f}}\,\eta\wedge(\d\eta)^{n}\right|_{h^{-1}(0)}=\left.\frac{f(0)}{\mathcal Z_{f}}\,\eta\wedge(\d\eta)^{n}\right|_{h^{-1}(0)}
	\eeq
	is an invariant measure along the orbits of the flow \eqref{z1}-\eqref{z3} lying on the level surface $h^{-1}(0)$. Here $f(h)$
	is any function of the contact Hamiltonian only and $\mathcal Z_{f}$ is a normalization factor, which in general depends on $f(h)$.
	\end{Theorem}
	
	\begin{proof} 
	The proof follows from eqs.  \eqref{evolutionfh} and  \eqref{LieDerVolume}, which imply that both $f(h)$ and the volume form are invariant on the level surface $h^{-1}(0)$.
	 Finally, since the distribution is uniform, we call it microcanonical.
	\end{proof}

{Let us consider now the orbits of the flow that extend outside the special surface $h^{-1}(0)$.
In the following theorem we provide the unique invariant measure for such orbits.
	}

	\begin{Theorem}[Canonical invariant measure on $\mathcal T\setminus h^{-1}(0)$]\label{Theoremneq0}
	For any nonconservative system \eqref{z1}-\eqref{z3} given by the corresponding contact Hamiltonian $h(S,q^{a},p_{a})$, the measure
	\beq\label{CanInvariantMeasure}
	\d\mu\equiv \frac{|h|^{-(n+1)}}{\mathcal Z}\,\eta\wedge(\d\eta)^{n}
	\eeq
	is an invariant measure of the flow \eqref{z1}-\eqref{z3} along the orbits  lying outside of the level surface $h^{-1}(0)$. Here $\mathcal Z$ is a normalization factor. Moreover, such measure
	is the unique invariant measure for such orbits {whose probability density  with respect to the standard volume form} depends only on $h$.
	\end{Theorem}

	\begin{proof} We start by calculating 
	\beq
	\begin{split}
	&\pounds_{X_{h}}\left(\rho\,\eta\wedge(\d\eta)^{n}\right)=\\
	&=\left(\pounds_{X_{h}}\rho\right)\eta\wedge(\d\eta)^{n}+\rho\,\pounds_{X_{h}}\left(\eta\wedge(\d\eta)^{n}\right)\\
	&=\left[X_{h}(\rho)
	+(n+1)\rho\,\xi(h)\right]\eta\wedge(\d\eta)^{n}\,,\\
	\end{split}
	\eeq
	where the second equality follows from \eqref{LieDerVolume}. 
	Now, assuming that $\rho=\rho(h)$, {it follows from \eqref{evolutionfh} that $X_{h}(\rho)=h\,\rho'(h)\,\xi(h)$} and therefore one is left with 
	the differential equation for $\rho(h)$ given by 
	\beq
	\frac{\d\rho}{\d h}=-(n+1)\frac{\rho}{h}\,,
	\eeq
	whose only solution is the probability density in \eqref{CanInvariantMeasure}, 
	where the absolute value is needed to guarantee that the probability density is non-negative. 
	\end{proof}

\section{Highlights and future directions}

To highlight the importance {of eqs. \eqref{MicroInvariantMeasure} and \eqref{CanInvariantMeasure},
notice 
that these invariant measures on the phase space of nonconservative contact Hamiltonian systems are 
the  counterparts of the microcanonical and the canonical measures for conservative systems respectively.
Indeed, \eqref{MicroInvariantMeasure} gives an invariant measure along the orbits of the flow that conserve the contact Hamiltonian,
while \eqref{CanInvariantMeasure} is the analogue of the standard Gibbs' canonical measure
	\beq\label{Gibbsmeasure}
	\d\mu_{\rm Gibbs}\equiv \frac{\exp(-H)}{\mathcal Z_{Gibbs}}\,(\d\alpha)^{n},
	\eeq
where $\alpha$ is given in \eqref{taut}.
The differences between the two measures \eqref{CanInvariantMeasure} and \eqref{Gibbsmeasure}
are  in the dimension
of the phase spaces and in the functional dependence with respect to the corresponding Hamiltonian functions. The use of extended phase
spaces has already been proven 
to be useful for the derivation of non-Hamiltonian dynamics \cite{1980JChPh722384A,mukunda1974classical,hoover1985canonical,1984Evans,1998Morriss,evansbook}.
As for the different functional form in \eqref{CanInvariantMeasure} and \eqref{Gibbsmeasure},}
in the Gibbs
case we have the well-known exponential dependence, while in the contact case there is a power law dependence.
{Power law} distributions are ubiquitous in the science
of collective phenomena, 
 emerging as the `slowly decaying' counterpart of Gibbs' distributions.
{Moreover, they can also be derived as the equilibrium distributions in the maximum entropy approach, provided that data are given 
 in terms of the average value of the logarithm of the relevant observables \cite{visser2013zipf}.}
It is worth also remarking that  here the number $n$ is the number of degrees of freedom in the system. Therefore in large systems the exponent in
 \eqref{CanInvariantMeasure} is extremely large and the distribution is well approximated by \eqref{Gibbsmeasure}.
However eq. \eqref{CanInvariantMeasure} in principle is valid in any dimension. Therefore it could be appropriate in order
 to extend Gibbsian statistical mechanics to small systems,
which are known to present power law distributions, rather than exponential ones.   
{For this reason} we call the invariant measure \eqref{CanInvariantMeasure} the \emph{canonical measure for nonconservative contact Hamiltonian systems}.
In this sense, we interpret the normalizing factor $\mathcal Z$ in \eqref{CanInvariantMeasure} as the \emph{canonical partition function} for nonconservative contact Hamiltonian 
systems.

{We have argued that the flow itself splits into two different types of orbits, those with $h=0$ and those with $h\neq 0$. 
In the thermodynamic context this property has been used to define both quasi-static and relaxation processes respectively \cite{Mrugaa:2000aa,2014Bravc,goto2014legendre}.
This fact could be interesting when considering continuous phase transitions, which are characterized by a power law, scale-invariant, behavior.
In our formalism the passage from the standard equilibrium thermodynamic representation to the critical region of a phase transition might be interpreted as the
passage from processes living on the $h^{-1}(0)$ surface -- for which the invariant measure is the microcanonical measure \eqref{MicroInvariantMeasure} --
to processes that lie outside this surface and for which the invariant measure has a power law distribution in terms of the relevant contact Hamiltonian function $h$. 
This aspect has not been addressed here and we defer it to further developments.}

At this point, we provide a simple example 
in which this formalism {can  be used to generalize} standard symplectic statistical mechanics.
Consider the dissipative system \eqref{dissipative}. Such a simple system cannot be described by the means of symplectic Hamiltonian
motion. On the contrary,  if we assume the phase space to have a contact geometry, then we can define the Hamiltonian 
$h(S,q^{a},p_{a})\equiv H(q^{a},p_{a})-\alpha S$, where $H(q^{a},p_{a})$ is the standard mechanical Hamiltonian and $\alpha$ is a constant. 
With this choice, the equations
of motion of the contact flow \eqref{z1}-\eqref{z3} give exactly the dissipative system \eqref{dissipative} (plus an extra equation for the evolution
of the `generalized principal function' $S$).
The further investigation of this system and other nonconservative contact Hamiltonian systems will be the subject of future work.
Moreover, it will be also of  interest to compare our results with Hamilton's principle for nonconservative systems derived in \cite{GalleyPRL}.   
We expect that this analysis will shed light over the meaning of $S$ in the nonconservative case.

{To conclude, in this work we have given the basic mathematical landscape for the statistical mechanics of a class of nonconservative systems.
Our results  generalize the symplectic description of conservative systems.
We have  focused  on nonconservative systems whose phase space has a contact geometry --  contact Hamiltonian systems -- 
and we have derived the corresponding equations of motion. The central result of this  work is the 
fact that we can provide  both a microcanonical and a canonical invariant measure along the flow of all such systems.}
Moreover, we have proved that  the canonical measure is unique and has a power law distribution (see Theorem~\ref{Theoremneq0}).
Our results thus open the possibility to understand the statistical mechanics of nonconservative systems from a {new and formal perspective.}
Moreover, they could be useful in the construction of robust 
and efficient algorithms such as Molecular Dynamics or Markov Chain Monte Carlo, 
{akin to the use} of symplectic integrators for conservative systems \cite{2014arXiv1410.5110B} and therefore they are potentially relevant for 
{future developments in pure and applied sciences.}

\section*{Acknowledgements}
The authors would like to thank C. S. L\'opez-Monsalvo and H. Quevedo for insightful discussions and the anonymous referees for their helpful suggestions.
AB acknowledges financial support from the A. della Riccia Foundation (postdoctoral fellowship). DT acknowledges financial support from CONACyT, CVU No. 442828.

\bibliographystyle{unsrt}
\bibliography{GTD}

\end{document}